\documentclass[smallabstract,smallcaptions]{dccpaper}
\usepackage[T1]{fontenc}
\usepackage[utf8]{inputenc}
\usepackage{bm}
\usepackage{amsmath}
\usepackage{amssymb,amsthm}
\usepackage{mathtools}
\usepackage{hyperref}
\usepackage[nameinlink,capitalise]{cleveref}
\usepackage{color}
\usepackage{url}

\newlength{\figurewidth}
\newlength{\smallfigurewidth}
\newtheorem{theorem}{Theorem}
\newtheorem{definition}{Definition}
\newtheorem{corollary}{Corollary}
\newtheorem{lemma}{Lemma}

\usepackage{graphicx}
\graphicspath{{figures/}}

\usepackage[backend=biber, style=ieee, natbib=true, mincitenames=1, maxcitenames=2, giveninits=true, uniquename=init]{biblatex}
\AtBeginDocument{\toggletrue{blx@useprefix}}
\AtBeginBibliography{\togglefalse{blx@useprefix}}

\usepackage{tikz}
\usetikzlibrary{backgrounds, calc, chains, fit, matrix, positioning, shadows, shapes, circuits.ee}
\tikzset{input/.style={}}
\tikzset{output/.style={}}
\tikzset{op/.style={circle, draw, thick, fill=black!10, minimum size=2.5ex, inner sep=0ex}}
\tikzset{filter/.style={rectangle, draw, thick, fill=black!10, minimum size=3.5ex, inner sep=1ex}}
\tikzset{nn/.style={trapezium, trapezium angle=80, draw, thick, fill=black!10, inner sep=1ex}}
\tikzset{branch/.style={circle, draw, thick, fill=black, minimum size=.5ex, inner sep=0ex}}
\tikzset{tensor/.style={rectangle, draw, thick, fill=white, minimum size=2em, double copy shadow={shadow xshift=.5ex,shadow yshift=-.5ex}}}
\tikzset{image/.style={rectangle, draw, thick, fill=white, minimum size=2em}}
\tikzset{>=direction ee}

\usepackage{pgfplots}
\pgfplotsset{compat=1.14}
\pgfplotsset{every axis/.append style={enlargelimits={abs=3pt},grid,axis lines=left}}
\pgfplotsset{every axis plot/.append style={thick,mark size=1.5pt,line join=bevel,mark options={solid}}}
\pgfplotsset{label style={font=\small}}
\pgfplotsset{tick label style={font=\footnotesize}}
\pgfplotsset{grid style={color=black!10}}
\pgfplotsset{legend style={draw=none,opacity=.85,font=\footnotesize,cells={anchor=west,opacity=1}}}
\pgfplotsset{every non boxed x axis/.style={xtick align=center,shorten <=-.5\pgflinewidth}}
\pgfplotsset{every non boxed y axis/.style={ytick align=center,shorten <=-.5\pgflinewidth}}
\pgfplotsset{every non boxed z axis/.style={ztick align=center,shorten <=-.5\pgflinewidth}}
\pgfplotsset{/pgf/number format/1000 sep={\,}}

\definecolor{gblue}{HTML}{1f77b4}
\definecolor{ggreen}{HTML}{2ca02c}

\DeclareMathOperator{\lce}{LCE}
\DeclareMathOperator{\unif}{unif}
\DeclareMathOperator{\atan}{tan^{-1}}

\DeclarePairedDelimiterX{\divergence}[2]{[}{]}{#1\;\delimsize\|\;#2}

\newcommand{\1}{\ensuremath{\mathbf{1}}}

\newcommand{\N}{\operatorname{\mathbb N}}

\renewcommand{\P}{\operatorname{Pr}}

\newcommand\comment[1]{}
\newcommand{\Ltwo}{\ensuremath{L^2[0,1]}}
\newcommand{\conv}{\star}

\newcommand*{\JOURNAL}{}%

\begin{document}

\title{Neural Networks Optimally Compress the Sawbridge}

\author{%
Aaron B. Wagner$^{\ast}$ and Johannes Ballé$^{\dag}$\\[0.5em]
{\small\begin{minipage}{\linewidth}\begin{center}
\begin{tabular}{ccc}
$^{\ast}$Cornell University & \hspace*{0.5in} & $^{\dag}$Google Research \\
388 Frank H.T.~Rhodes Hall && 1600 Amphitheatre Pkwy \\
Ithaca, NY 14853 USA && Mountain View, CA 94043 USA\\
\url{wagner@cornell.edu} && \url{jballe@google.com}
\end{tabular}
\end{center}\end{minipage}}
}

\maketitle

\begin{abstract}
    Neural-network-based compressors have proven to be remarkably effective
at compressing sources, such as images, that are nominally 
high-dimensional but presumed to be concentrated on a low-dimensional 
manifold. We consider a continuous-time random process
that models an extreme version of such a source,
wherein the realizations fall along a one-dimensional ``curve''
in function space that has infinite-dimensional linear span. We
precisely characterize the optimal entropy-distortion tradeoff
for this source and show numerically that it is achieved by
neural-network-based compressors trained via stochastic gradient
descent. In contrast, we show both analytically and experimentally
that compressors based on the classical Karhunen-Lo\`{e}ve transform 
are highly suboptimal at high rates.

\end{abstract}
\vspace{1.5em}


\section{Introduction}

Artificial Neural-Network (ANN)-based compressors have recently achieved 
notable successes on the task of lossy compression of multimedia, spanning
an array of sources and in some cases outperforming compressors
that have been extensively optimized~(see, e.g.,~\cite{BaChMiSiJo20} 
and the references therein). 

One explanation for the exemplary performance of ANNs is that it derives 
from their ability to approximate arbitrary functions 
(e.g.,~\cite{LeLiPiSc93}), which in turn enables them to perform 
nonlinear dimensionality reduction~\citep{HiSa06}. To see this, first 
consider classical rate--distortion theory for Gaussian sources,
which is based on linear dimensionality reduction~\cite[Sec.~4.5.2]{Berger:RD}. 
Specifically, one projects the source realization onto an orthogonal 
family of reconstructions obtained from the Karhunen-Lo\`{e}ve
Transform (KLT) of the source. One then quantizes the resulting 
coefficients, say with a uniform quantizer followed by entropy 
coding~\cite[Sec.~5.5]{Pearlman:Said}.
At the decoder, the inverse transform is applied to the quantized
coefficients. The size of the orthogonal family is generally less than
the dimensionality of the source, which provides some amount of compression. 
The quantization process provides more. For Gaussian sources, this 
architecture is provably near-optimal at high 
rates~\cite[Sec.~5.6.2]{Pearlman:Said}. 
In particular, using a nonlinear transform in place of the KLT
provides essentially no benefit. 

For real-world multimedia sources, however,
there is reason to believe that allowing for nonlinear transforms 
would be advantageous.  The distribution
of natural images is widely suspected to be supported by a 
low-dimensional manifold in pixel-space~(e.g., \cite{HeBaRaSi15}),
for instance.
That is, while the linear span of the manifold may be high, 
there exists a continuous, presumably nonlinear, map with
a continuous, presumably nonlinear, inverse, from
the manifold to a low-dimensional Euclidean space. One
could in principle use such a map in place of the linear
projections in the classical architecture, with the reduced
dimensionality of the output afforded by allowing for nonlinear
transforms translating to a lower bit-rate.
State-of-the-art ANN-based compressors indeed follow this 
architecture, with nonlinear analysis and synthesis transforms 
surrounding a conventional uniform quantizer with entropy 
coding~\citep{BaChMiSiJo20}. 
As a consequence, one might hypothesize that ANNs are particularly
adept at compressing sources for which there is a large discrepancy
between the amount of dimensionality reduction afforded by linear
versus nonlinear transforms.

We provide evidence for this hypothesis by showing that
ANNs optimally compress a prototypical source of this type,
and that their performance beats the linear-transform-based approach
by a large margin.
We consider a particular stochastic process over $[0,1]$ that can be 
constructed from a continuous, nonlinear transformation of a single random
variable.
We learned of this process and its usefulness as a test-case for
compression algorithms from the survey by
Donoho~\emph{et al.}~\citep{DoVeDeDa98}, who in turn credit 
Meyer~\citep{Meyer:Ramp}.
Donoho~\emph{et al.}\ refer to this process as the \emph{Ramp}; we shall 
call it the \emph{sawbridge.} We focus on this process because it exposes
the largest possible gap between linear and nonlinear dimensionality
reduction. Donoho~\emph{et al.}\ point out that the sawbridge has the
same autocorrelation, and thus the same KLT,
as the Brownian bridge.\footnote{The two processes
also share the property that they start and end at zero, which
motivates our choice of the former's name.} In particular, the
KLT of the sawbridge is infinite-dimensional.  We show that 
any linear transform requires infinitely many components to 
recover the source. In contrast, there is a simple nonlinear 
transform that can recover the source realization from a 
one-dimensional projection.

We show that this discrepancy in dimensionality reduction
translates to a large gap in compression performance.
We analytically characterize the performance of optimal one-shot
compression for the sawbridge and numerically 
show that it is realized by a deep ANN
trained via stochastic gradient descent. We also characterize
the performance of KLT-based schemes and show, both numerically
and mathematically, that they are exponentially suboptimal.

The next section introduces the sawbridge process and its properties.
\cref{sec:optimal} shows how to optimally compress the sawbridge.
\cref{sec:linear} characterizes the performance of schemes based
on the KLT and entropy coding. \cref{sec:neural} numerically shows 
that the sawbridge 
is optimally compressed by existing neural network architectures
and training methods, and that linear-based methods are highly
suboptimal. All proofs have been omitted due to space
constraints but are available in the extended version of the
paper~\cite{WaBa:Sawbridge:Long}.

\section{The Sawbridge}
The focus of this paper is the following stochastic process.

\begin{figure}
\centering
\includegraphics{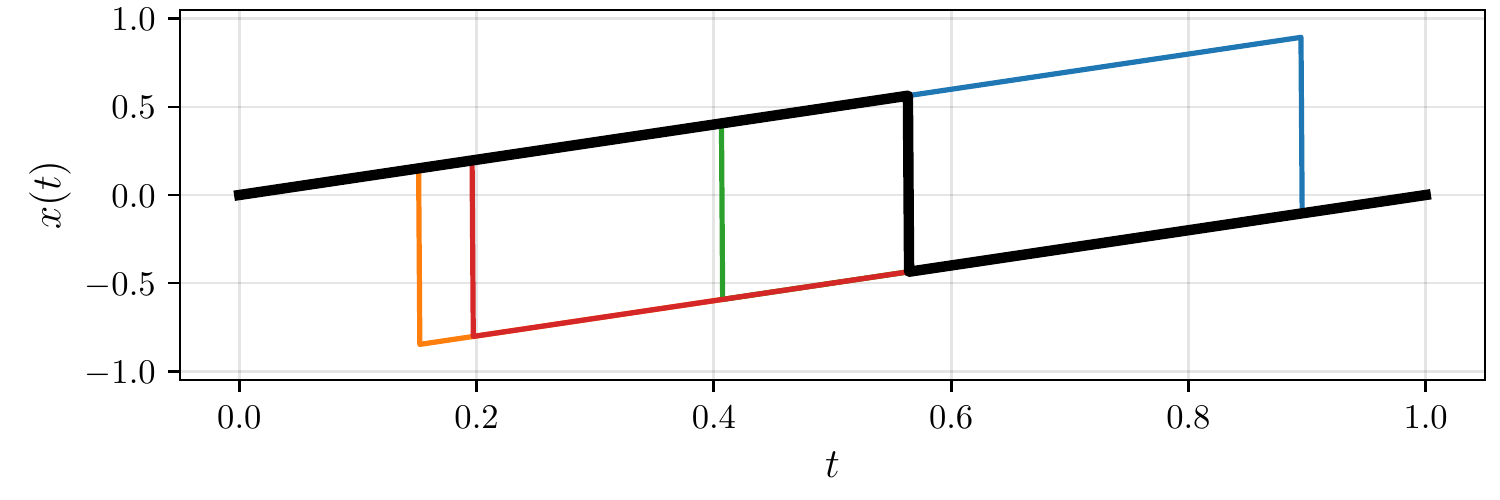}
\caption{Realizations of the sawbridge. The bold line represents one full realization; others show additional samples.}
\label{fig:realizations}
\end{figure}

\begin{definition}[cf.~\cite{DoVeDeDa98}]
    The \emph{sawbridge} is the process
    \begin{equation}
    X(t) = t - \1(t \ge U) \quad t \in [0,1],
        \label{eq:sawbridge}
    \end{equation}
    where $U$ is uniformly distributed over $[0,1]$. We use $X$
    or $X(\cdot)$ to refer to the entire process $\{X(t)\}_{t = 0}^1$.
\end{definition}

See \cref{fig:realizations} for sample realizations. In words, the process
jumps from the ``rail'' $t$ to the ``rail'' $t-1$
at the random time $U$. Alternatively, one can view
$-X$ as the centered empirical cumulative distribution function
of a single $\unif[0,1]$ random variable.

As noted in the introduction,
we are interested in the sawbridge because it exposes the largest
possible gap between the performance of linear and nonlinear
dimensionality reduction. Call a map $f : \Ltwo \mapsto \mathbb{R}^k$ a
\emph{transform (of dimension $k$)} if it is continuous and there exists
a continuous function $g : \mathbb{R}^k \mapsto \Ltwo$  (the
\emph{inverse transform}) such that
\begin{equation}
    \label{eq:lossless}
    g(f(X)) = X  \quad \text{a.s.}
\end{equation}
We call $\mathbb{R}^k$ the \emph{latent space}.

If $f$ and, especially, $g$ are permitted to be nonlinear,
then a transform of dimension one exists, which is clearly
the lowest possible.  Specifically, the choices
\begin{equation}
    f(x) = \int_{0}^1 x(t) \; dt
    \label{eq:optlinearf}
\end{equation}
and
\begin{equation}
    (g(y))(t) = t - \1(y \le t - 1/2)
\end{equation}
suffice. This comports with the intuition
that the process is completely described by $U$. In fact,
the $f(\cdot)$ in (\ref{eq:optlinearf}) satisfies $f(X) = U - 1/2$ a.s.
if $X$ and $U$ are related by~(\ref{eq:sawbridge}).
Note that this $f$ happens to be a linear map.

On the other hand, if $f$ and $g$ are both required to be linear
maps then~(\ref{eq:lossless}) is impossible for any finite $k$.
This follows from the following result on the
Karhunen-Lo\`{e}ve expansion of the process.
Note that the sawbridge is zero mean and let
\begin{equation}
    K(s,t) = E[X(s) X(t)] = \min(s,t) - st
\end{equation}
denote its autocorrelation.

\begin{theorem}
    The functions
    \begin{equation}
        \phi_k(t) = \sqrt{2} \cdot \sin(\pi k t) \quad t \in [0,1], \quad
                      k = 1, 2, \ldots,
    \end{equation}
    form an orthonormal basis for $\Ltwo$. They are eigenfunctions
    of $K(\cdot,\cdot)$ with corresponding eigenvalues
    \ifdefined\JOURNAL
    \begin{equation}
        \lambda_k = \frac{1}{\pi^2 k^2},
    \end{equation}
    \else
    $\lambda_k = \frac{1}{\pi^2 k^2}$,
    \fi
    meaning that
    \begin{equation}
        \label{eq:eigenfunction}
        \int_{0}^1 K(s,t) \phi_k(t) \; dt = \lambda_k \cdot \phi_k(s)
    \end{equation}
    for all $k$ and $s \in [0,1]$. If we define
    \begin{align}
           \label{eq:KLTcoeffdef}
        Y_k & := \int_0^1 X(t) \phi_k(t) \; dt \\
           & = - \sqrt{2 \lambda_k} \cos(\pi k U),
           \label{eq:explicitKLTcoeff}
    \end{align}
    then $\{Y_{k}/\sqrt{\lambda_{k}}\}_{k = 1}^\infty$ is a
       sequence of uncorrelated, zero-mean,
      unit-variance random variables such that the sawbridge
    can be written as
    \begin{equation}
        \label{eq:invKLT}
        X(t) = \sum_{k = 1}^\infty Y_{k} \phi_k(t),
    \end{equation}
    \ifdefined\JOURNAL
   in the sense that 
    \begin{equation}
        \lim_{n \rightarrow \infty} \sup_{0 \le t \le 1} E\left[
               \left(X(t) - \sum_{k = 1}^n Y_k \phi_k(t)\right)^2\right]
                   \; dt = 0.
    \end{equation}
    \else
    where the convergence in \cref{eq:invKLT} is in quadratic mean, uniformly
    over $t$.
    \fi
\end{theorem}

\ifdefined\JOURNAL
\begin{proof}
    It is straightforward to verify that $\{\phi_k(\cdot)\}_{k = 1}^\infty$
    forms an orthonormal family satisfying~(\ref{eq:eigenfunction}).
      To show that it
    forms an orthonormal  
       basis for $\Ltwo$, consider any $g(\cdot)$ in $\Ltwo$
    such that
    \begin{equation}
        \label{eq:basishypothesis}
        \int_0^1 g(t) \sin(\pi k t) \; dt = 0 \ \text{for all $k$.}
    \end{equation}
    Consider the odd extension of $g(\cdot)$ to $[-1,1]$,
    \begin{equation}
        g_{o}(t) = \begin{cases}
            - g(-t) & \text{if $t \in [-1,0]$} \\
            g(t) & \text{if $t \in (0,1]$}.
        \end{cases}
    \end{equation}
    From (\ref{eq:basishypothesis}), we have
    \begin{equation}
        \int_{-1}^1 g_{o}(t) \sin(\pi k t) \; dt = 0 \quad k = 1,2,3,\ldots.
    \end{equation}
    At the same time, since $g_{o}(\cdot)$ is odd, we have
    \begin{equation}
        \int_{-1}^1 g_{o}(t) \cos(\pi k t) \; dt = 0, \quad k = 0,1,2,\ldots.
    \end{equation}
    The completeness of the Fourier basis implies
    that $g(\cdot) \equiv 0$
    and hence $\{\phi_k(\cdot)\}_{k = 1}^\infty$ is also an orthonormal
    basis~(e.g.,~\cite[Theorem~8.16]{Rudin:RealAnalysis},
       \cite[Theorem 5.27]{Folland:RealAnalysis}; note that the 
       set of Riemann integrable functions is dense in 
       $L^2$~\cite[Prop.~6.7]{Folland:RealAnalysis}).
    The rest of the conclusion then follows from the Karhunen-Lo\`{e}ve
    theorem~(e.g.,~\cite[Prop.~4.1]{Wong:SPES}; note that the 
    sawbridge is quadratic-mean continuous, even though its
    realizations are obviously not continuous),
     except for (\ref{eq:explicitKLTcoeff}), which follows
    by integrating (\ref{eq:KLTcoeffdef}) by parts.
\end{proof}
\else
\fi

This implies that if the transform and inverse must be linear,
then an infinite number of dimensions is required to represent the 
sawbridge, as shown in the following corollary.

\begin{corollary}
    For any finite $k$, if $f : \Ltwo \mapsto \mathbb{R}^k$ and
      $g : \mathbb{R}^k
       \mapsto \Ltwo$ are linear maps 
         then they cannot satisfy~(\ref{eq:lossless}).
\end{corollary}

\ifdefined\JOURNAL
\begin{proof}
    Suppose $f$ and $g$ are linear and let $Y = g(f(X))$.
    Then we can write
    \begin{equation}
        Y(t) = \sum_{i = 1}^k A_i \psi_i(t)
    \end{equation}
    for some random variables $A_1, \ldots, A_k$ and orthonormal functions
    $\psi_1(\cdot), \ldots, \psi_k(\cdot)$. Then the autocorrelation of
    $Y$ has at most $k$ nonzero eigenvalues and hence $X$
    and $Y$ have distinct distributions.
\end{proof}
\else
\fi

The large gap between linear and nonlinear dimensionality reduction
for the sawbridge has consequences for compression, to which we turn next.

\section{Optimal Compression}
\label{sec:optimal}
We consider a one-shot form of compression in which 
the goal is to minimize the entropy of the compressed representation
for a given mean squared error.

\begin{definition}
    An \emph{encoder} is a map
    $f : \Ltwo \mapsto \N$.
    Its \emph{entropy} and \emph{distortion} are
    \begin{align*}
        H(f) & = - \sum_{i \in \N} \P(f(X) = i) 
                \log \P(f(X) = i) 
               \\
        D(f) & = E\left[\int_0^1 (X(t) - E[X(t)|f(X)])^2 \; dt\right],
    \end{align*}
    respectively.\footnote{Throughout, all logarithms are base two.}
    \label{def:optenc}
\end{definition}

Note that
Definition~\ref{def:optenc} assumes that the reproduction
$E[X(t)|f(X)]$ need not be a valid sawbridge realization.
Also note that an encoder is distinct from a transform in that it maps
the source realization to a discrete set and is therefore not invertible.
In practice, compression involves mapping the source realization
to a variable-length bit string, whose
expected length one might wish to minimize. The minimum such
expected length, $L^*(f)$, is known to satisfy
    $H(f) \le L^*(f) < H(f) + 1$
if one requires that the codewords are 
prefix-free~(e.g.,~\cite[Theorem~5.4.1]{Cover:IT2}) and 
\begin{align}
    L^*(f) & \le H(f) \\
    H(f) & \le L^*(f) + (1+L^*(f))\log (1+L^*(f)) - L^*(f) \log L^*(f),
\end{align}
if one does not~\cite[Theorem 1]{Szpankowski:NoPrefix}. 
As such, it is reasonable to
focus on $H(f)$ as the figure-of-merit, especially at high rates.

\begin{definition}
    The \emph{entropy-distortion function} of the sawbridge is
    \begin{equation}
    H(\Delta) = \inf_f H(f),
    \end{equation}
where the infimum is over all encoders $f$ such that
    $D(f) \le \Delta$.
\end{definition}

\begin{theorem}
    If $\Delta \ge 1/6$, then $H(\Delta) = 0$.
For any $0 < \Delta < 1/6$, we have
    \begin{equation}
    H(\Delta) = - \left\lfloor \frac{1}{p} \right\rfloor \cdot p \log p -
       q \log q,
   \end{equation}
   where $q = \left( 1- \left\lfloor \frac{1}{p} \right\rfloor \cdot p\right)$
    and $p$ is the unique number in $(0,1)$ such that 
    \begin{equation}
\left\lfloor \frac{1}{p} \right\rfloor \cdot p^2 + 
  q^2 = 6\Delta.
    \end{equation}
    \label{thm:optimal}
\end{theorem}

\ifdefined\JOURNAL
\begin{proof}
    If $\Delta \ge 1/6$, then a trivial encoder with a singleton
    range achieves $D(f) \le \Delta$ and $H(f) = 0$. Suppose
    $0 < \Delta < 1/6$.
    Each realization of $X$ can be identified with a unique
    realization of $U$, so the realizations of $X$
    are in one-to-one 
    correspondence with $[0,1]$. Thus for any encoder $f$ and
    any $i$, $f^{-1}(i)$ can be identified with a subset of
    $[0,1]$, say $A_i$, such that
    \begin{align*}
        & E\left[\int_0^1 (X(t) - E[X(t)|f(X)])^2 \; dt
               \middle|f(X) = i\right] \\
               & = E\left[\int_0^1 (X(t) - E[X(t)|U \in A_i])^2 \; dt
                 \middle|U \in A_i\right].
    \end{align*}
    Now consider an arbitrary $A \subset [0,1]$ such that $\mu(A) > 0$,
    where $\mu(\cdot)$ is the Lebesgue measure. We will show the
    inequality
    \begin{equation}
        \label{eq:distortionLB}
        E\left[\int_0^1 (X(t) - E[X(t)|U \in A])^2 \; dt
                  \middle|U \in A\right] \ge
         \frac{\mu(A)}{6}.
    \end{equation}
    To this end, note that we can write
    \begin{align}
        &  E\left[\int_0^1 (X(t) - E[X(t)|U \in A])^2 \; dt\middle|U \in A\right] \\
        & = E\left[\int_0^1 (\1(U \le t) - E[\1(U \le t)|U \in A])^2 \; dt\middle|U \in A\right] \\
        & = \int_{0}^{1} \P(U \le t|U \in A)
               (1 - \P(U \le t|U \in A)) \; dt.
               \label{eq:varform}
        \end{align}
    First suppose that $A$ is a union of intervals,
    \begin{equation}
    A = \cup_{i = 1}^k [a_i,b_i],
    \end{equation}
    where $0 \le a_1 < b_1 \le a_2 < b_2 \le \ldots \le a_k < b_k \le 1$.
    Define
    \begin{equation}
  \beta_0 = 0 \quad \text{and} \quad \beta_i = \sum_{\ell = 1}^i \frac{b_\ell -
          a_\ell}{\mu(A)}
    \end{equation}
    and note that $\beta_k = 1$.  Then from (\ref{eq:varform}) we have
    \begin{align}
        &  E\left[\int_0^1 (X(t) - E[X(t)|U \in A])^2 \; dt
             \middle|U \in A\right] \\
        & \ge \sum_{\ell = 1}^k \int_{a_{\ell}}^{b_{\ell}} \P(U \le t|U \in A)
             (1 - \P(U \le t|U \in A)) \; dt \\
       & = \sum_{\ell = 1}^k \int_{a_\ell}^{b_\ell} \left(\beta_{\ell-1} +
                 \frac{t - a_{\ell}}{\mu(A)} \right)
                 \left( 1- \beta_{\ell - 1} - \frac{t - a_{\ell}}
                    {\mu(A)} \right) \; dt.
                    \label{eq:needvarchange}
    \end{align}
    Performing the change of variable 
    \begin{equation}
        s = \beta_{\ell - 1} + \frac{t - a_{\ell}}{\mu(A)},
    \end{equation}
    the quantity in~(\ref{eq:needvarchange}) equals
    \begin{align}
        \sum_{\ell = 1}^k \int_{\beta_{\ell-1}}^{\beta_{\ell}} 
                     s(1-s) \; ds \cdot \mu(A) 
             & = \int_0^1 s(1-s) \; ds \cdot \mu(A) \\
             & = \frac{\mu(A)}{6}.
    \end{align}
    To handle the general case, let $\{A_n\}_{n = 1}^\infty$ be a sequence
    of sets, each of which is a finite union of intervals, such 
    that~\cite[Theorem~1.20]{Folland:RealAnalysis}
    \begin{equation}
        \lim_{n \rightarrow \infty} \mu(A \cap {A}_n^{c}) + 
                  \mu(A_n \cap A^c) = 0.
    \end{equation}
    Then for each $t$,
    \begin{align}
        \lim_{n \rightarrow \infty} \P(U \le t|U \in A_n)
            & = \lim_{n \rightarrow \infty} \frac{\mu([0,t] \cap A_n)}
                        {\mu(A_n)} \\
                        & = \frac{\mu([0,t] \cap A)}{\mu(A)} \\
                        & = \P(U \le t|U \in A).
    \end{align}
    It follows by dominated convergence that
    \begin{align}
        & \int_{0}^1 P(U \le t|U \in A) ( 1- \P(U \le t|U \in A)) \; dt \\
        & = \lim_{n \rightarrow\infty}
          \int_{0}^1 P(U \le t|U \in A_n) ( 1- \P(U \le t|U \in A_n)) \; dt \\
          & \ge \lim_{n \rightarrow \infty} \frac{\mu(A_n)}{6}
              = \frac{\mu(A)}{6},
    \end{align}
    which establishes~(\ref{eq:distortionLB}) and further implies
    that (\ref{eq:distortionLB}) holds with equality if $A$ is an interval.
     Applying this to the distortion,
    \begin{align*}
        D(f)  & =  E\left[\int_0^1 (X(t) - E[X(t)|f(X)])^2 \; dt\right]  \\
       & = \sum_{i = 1}^\infty \P(U \in A_i) 
        E\left[\int_0^1 (X(t) - E[X(t)|U \in A_i])^2 \; dt\middle|
                 U \in A_i \right] \\
        & \ge \sum_{i = 1}^\infty \P(U \in A_i) \cdot \frac{\mu(A_i)}{6} \\
        & = \sum_{i = 1}^\infty \frac{\mu(A_i)^2}{6},
    \end{align*}
    with equality if all of the (nonempty) $A_i$ are intervals.
    Since $H(f)$ depends on the $\{A_i\}$ only through $\{\mu(A_i)\}$,
    it follows that we can restrict attention to encoders that quantize
    $U$ to intervals.
    Writing $p_i = \mu(A_i)$, we can then express the entropy-distortion
    tradeoff as 
    \begin{align*}
        H(\Delta) = 
        \inf_{\{p_i\}_{i = 1}^\infty} & - \sum_{i = 1}^\infty p_i \log p_i \\
        \text{subject to} \ & p_i \ge 0 \ \text{for all $i$} \\
                  & \sum_{i = 1}^\infty p_i = 1 \\
              & \sum_{i = 1}^\infty \frac{p^2_i}{6} \le \Delta.
    \end{align*}
    This is also the entropy-distortion tradeoff for the problem of
    quantizing a $\unif[0,1]$ random variable subject to an $L^1$ 
    distortion constraint of $3\Delta/2$, assuming that all of the
    quantization cells are intervals. The latter problem is solved
    by a more general result of Gy\"{o}rgy and 
    Linder~\cite{Gyorgy:Uniform}, from 
    which the conclusion follows.
\end{proof}
\else
\fi

A plot of $H(\Delta)$ is included in \cref{fig:entropy_distortion}. Note that,
unlike the rate-distortion function, the entropy-distortion
function is not guaranteed to be convex and indeed it is not
in this case.  It reaches its lower convex envelope, which
we denote by $\lce(H(\cdot))$, at points of the 
form $(\log M, 1/(6M))$~\cite[Corr.~6]{Gyorgy:Uniform}. 
These points are achieved by encoders that quantize $U$ to
one of several equal-sized intervals. In particular,
for any $\lambda > 0$, minimizers of the Lagrangian
\begin{equation}
    \label{eq:Lagrangian}
\min_f \ H(f) + \lambda \cdot D(f)
\end{equation}
are encoders of this type.  Likewise, 
since $\lce(H(\cdot))$ describes the
entropy-distortion tradeoff of randomized encoders,
the best randomized encoders are those that randomly 
toggle among deterministic encoders that uniformly quantize $U$.

\ifdefined\JOURNAL
Although Theorem~\ref{thm:optimal} was motivated by a desire to
characterize the best variable-rate encoders, it implies
the following characterization of the best fixed-rate encoders.

\begin{corollary}
    Define an \emph{$M$-encode} for the sawbridge as an encoder with
    the property that the support of $f(X)$ has cardinality
    $M$ or less. The minimum distortion among all $M$-codes is
    $\frac{1}{6M}$, which is achieved by an encoder that quantizes
    $U$ uniformly to $M$ different values. 
\end{corollary}

\begin{proof}
  The encoder $f$ that uniformly quantizes $U$ to $M$ cells achieves 
      $D(f) = 1/(6M)$. Any other encoder that quantizes $U$ into at
      most $M$ cells must have entropy at most $\log M$ and therefore,
      by Theorem~\ref{thm:optimal}, distortion at least $1/(6M$).
\end{proof}
\else
\fi

Theorem~\ref{thm:optimal} also implies a simple high-rate characterization.

\begin{corollary}
    For the sawbridge,
    \begin{equation}
        \lim_{\Delta \rightarrow 0} \left|
            H(\Delta) - \log \frac{1}{6 \Delta} \right| = 0.
    \end{equation}
    \label{corollary:opthighrate}
\end{corollary}

\ifdefined\JOURNAL
\begin{proof}
    Define the function $M(\Delta)$ as $\lfloor \frac{1}{6 \Delta} \rfloor$.
    Then $\frac{1}{6(M(\Delta)+1)} < \Delta \le \frac{1}{6M(\Delta)}$, so
    by the monotonicity of $H(\cdot)$ and Theorem~\ref{thm:optimal},
    \begin{equation}
        \log M(\Delta) \le H(\Delta) < \log (M(\Delta)+1).
    \end{equation}
    Since $6\Delta M(\Delta) \rightarrow 1$, this implies the result.
\end{proof}
\else
\fi

Note that this corollary implies that $H(\Delta) / \log \frac{1}{6 \Delta}
\rightarrow 1$ but is significantly stronger. Next we shall see that
a compressor that follows the classical approach based on the 
KLT is far from meeting this optimal 
performance at high rates.

\section{KLT-Based Compression}
\label{sec:linear}
The classical approach to compressing a source such as the sawbridge
is to uniformly quantize and separately entropy-code a subset of the 
KLT coefficients.
Specifically, given a target distortion $\Delta$, define the constants
\ifdefined\JOURNAL
\begin{align}
    \label{eq:optD}
    D & = \frac{\Delta^2 \pi^2}{4} \\
    \label{eq:optk}
    K & = \left\lceil \frac{1}{\pi \sqrt{D}} \right\rceil \\
    \label{eq:optdelta}
    \delta & = \frac{\sqrt{12 \gamma}}{K}
\end{align}
\else
$D = \frac{\Delta^2 \pi^2}{4}$,
    $K = \left\lceil \frac{1}{\pi \sqrt{D}} \right\rceil$,
    and
    $\delta = \frac{\sqrt{12 \gamma}}{K}$,
\fi
where $\gamma > 0$ is the unique solution to the fixed-point equation
$    \atan(\pi \sqrt{\gamma}) = \frac{1}{\pi \sqrt{\gamma}}.$
Note that the dependence on $\Delta$ is suppressed in all three constants.
Consider the stochastic encoder $f_{\Delta}$ that quantizes the first
$K$ coefficients of the KLT to resolution $\delta$
using random dither. That is,
\begin{equation}
    f_{\Delta}(X) = \left(\left\lfloor \frac{ Y_{\ell}}{\delta} 
        + U_{\ell} \right\rceil, \ell \in \{1,\ldots,K\}\right),
\end{equation}
where $U_1,\ldots,U_K$ are i.i.d.\ $\unif[-1/2,1/2]$ and
$\lfloor \cdot \rceil$ denotes rounding to the nearest integer.
The $U_1,\ldots,U_K$ represent side randomness that is 
independent of the source and available when decoding.
We first show that $f_{\Delta}$ achieves distortion $\Delta$. 

\begin{lemma}
    The encoder $f_{\Delta}$ satisfies
        $D(f_{\Delta}) \le \Delta$
    for all $0 < \Delta < 1/6$.
\end{lemma}

\ifdefined\JOURNAL
\begin{proof}
    Consider the decoder that reproduces 
      $Y_{\ell}$ from the encoded representation and
      $U_{\ell}$ via
      \begin{equation}
          \hat{Y}_{\ell} :=  \begin{cases}
             \frac{\lambda_{\ell}}{\lambda_{\ell} + \delta^2/12}
          \left(\left\lfloor \frac{ Y_{\ell}}{\delta}
          + U_{\ell} \right\rceil - U_{\ell}\right) \cdot\delta
               & \text{if $\ell \le K$} \\
               0 & \text{if $\ell > K$}.
         \end{cases}
      \end{equation}
      That is, the decoder uses $U_{\ell}$ only to subtract
      the effect of the dither and otherwise uses a linear estimator.
      Dithered quantization is known to be equivalent to an additive-noise
      channel~\cite{Zamir:Bandlimited}:
      \begin{equation}
          \left(Y_{\ell}, \left(\left\lfloor \frac{Y_{\ell}}{\delta}
          + U_{\ell} \right\rceil - U_{\ell}\right) \cdot\delta \right)
            \stackrel{d}{=} \left(Y_{\ell}, Y_{\ell} + \delta \cdot
                U_{\ell}\right) \quad \ell \le K.
      \end{equation}
      Thus the distortion satisfies
      \begin{align}
          D(f_{\Delta}) & \le \sum_{\ell = 1}^\infty (\hat{Y}_{\ell}
               - Y_{\ell})^2 \\
               & = \sum_{\ell = 1}^{K} (\hat{Y}_{\ell}
               - Y_{\ell})^2 + 
                 \sum_{\ell = K+1}^\infty \lambda_{\ell} \\
          &  = \frac{1}{K^2} \sum_{\ell = 1}^{K}
                  \frac{\gamma}{1 + \gamma \pi^2 \ell^2/K^2}
                    + \sum_{\ell = K+1}^\infty \frac{1}{\pi^2 \ell^2} \\
         & \le \frac{1}{K} \int_0^1
                  \frac{\gamma}{1 + \gamma \pi^2 x^2} \; dx
                  + \int_{K}^\infty \frac{1}{\pi^2 x^2} \; dx.
      \end{align}
      Evaluating the two integrals (using the fact that $\atan(x)$ is
       the antiderivative of $1/(1+x^2)$)
      and using~(\ref{eq:optk}) and~(\ref{eq:optD}) gives the conclusion.
\end{proof}
\else
\fi

Since $f_{\Delta}(\cdot)$ is a stochastic encoder that
achieves distortion $\Delta$, it follows that
\begin{equation}
    \lce(H(\cdot))(\Delta) 
       \le H(f_{\Delta}|\{U_\ell\}) \le \sum_{\ell = 1}^{K} 
            H\left(\left\lfloor \frac{Y_{\ell}}{\delta} + U_{\ell}
                  \right\rceil \middle| U_{\ell}\right)
            =: \overline{H}(\Delta).
\end{equation}
Note that $\overline{H}(\Delta)$ is the rate that is achievable
if the quantized components are compressed separately, as is
typically done in practice.
When compressing a stationary Gaussian process over a long horizon
the analogue of the $\overline{H}(\Delta)$ bound is provably near-optimal
at high rates~(combining \citep[Sec.~5.6.2]{Pearlman:Said} and 
\citep[Sec~4.5.3]{Berger:RD}).
Comparing the following result with \cref{corollary:opthighrate}
shows that for the sawbridge, this bound is poor in the high-rate
regime.

\begin{theorem}
    Let $s(\cdot)$ 
    denote the arcsine density over $[-\sqrt{2},\sqrt{2}]$, i.e.,
    \begin{equation}
        s(x) = \frac{1}{\pi \sqrt{(\sqrt{2}-x)(\sqrt{2}+x)}},
        \label{eq:arcsine}
    \end{equation}
    and let $u_{x}(\cdot)$ denote the uniform density over $[-x/2,x/2]$.
    Then
    \begin{equation}
        \label{eq:KLTfinitek}
     \overline{H}(\Delta) = \sum_{\ell = 1}^K
     \left[ h(s(\cdot) \conv u_{\sqrt{12 \gamma} \pi \ell/K}(\cdot)) - 
             \log \frac{\pi \ell \sqrt{12 \gamma}}{K}
                \right],
     \end{equation}
     and as a result,
    \begin{equation}
        \label{eq:KLTrateconst}
        \lim_{\Delta \rightarrow 0} \Delta \cdot \overline{H}(\Delta) 
        = \frac{2}{\pi^2} \cdot \left(\int_0^1 h(s(\cdot) \conv 
            u_{\pi x \sqrt{12 \gamma}}(\cdot)) \; dx -
                     \log (\pi\sqrt{12 \gamma}/e) \right),
    \end{equation}
    where $h(\cdot)$ denotes differential entropy and $\conv$
    denotes convolution.
\end{theorem}

\ifdefined\JOURNAL
\begin{proof}
    Due to the dithering, the discrete entropy $H(\lfloor Y_{\ell}
        /\delta + U_{\ell} \rceil \ |U_{\ell})$ can be written as
    a mutual information~\cite[Theorem~1]{Zamir:Universal}
    \begin{align}
      H(\lfloor Y_{\ell} /\delta + U_{\ell} \rceil|U_{\ell})
        & = I(Y_{\ell}; Y_{\ell} + \delta U_{\ell}) \\
        & = h(Y_{\ell} + \delta U_{\ell}) - \log(\sqrt{12 \gamma}/K).
    \end{align}
    Let $Y$ have the arcsine distribution in~(\ref{eq:arcsine}). Then 
    $Y_{\ell}$ is identically distributed with $Y/(\pi \ell)$, so
    we have
    \begin{align}
      h(Y_{\ell} + \delta U_{\ell}) & = h(Y + \pi \ell \delta U_{\ell}) - \log(\ell \pi) \\ 
      & = h(s(\cdot) \conv u_{\sqrt{12 \gamma} \pi \ell/K}(\cdot)) - 
          \log(\ell \pi). 
    \end{align}
    This yields an exact expression for $\Delta \cdot \overline{H}(\Delta)$:
    \begin{align}
     \Delta \cdot \overline{H}(\Delta) & = \Delta \sum_{\ell = 1}^K
        \left[ h(s(\cdot) \conv u_{\sqrt{12 \gamma} \pi \ell/K}(\cdot)) - 
             \log \frac{\pi \ell \sqrt{12 \gamma}}{K}
              \right]  \\
             & = (K \Delta)  \sum_{\ell = 1}^K
            \left[ h(s(\cdot) \conv u_{\sqrt{12 \gamma} \pi \ell/K}(\cdot)) -
              \log \frac{\ell}{K} - 
              \log \sqrt{12 \gamma \pi^2} \right] \frac{1}{K} \\
          & \rightarrow \frac{2}{\pi^2} \cdot \left(\int_0^1 h(s(\cdot) \conv 
              u_{\pi x \sqrt{12 \gamma}}(\cdot)) \; dx -
                 \int_{0}^1 \log x \; dx -
                 \log (\sqrt{12 \gamma \pi^2}) \right),
    \end{align}
    where the convergence of the middle term follows from the monotone
    convergence theorem and the convergence of the first term
    follows from the dominated convergence theorem, which
    ensures that $h(s(\cdot) \conv u_y(\cdot))$ is continuous
    in $y$ for $y > 0$ 
    and hence Riemann integrable (since it is bounded if $y$ is bounded).
    The first equation establishes \cref{eq:KLTfinitek}.
    The antiderivative of $\log x$ is $x \log x - (\log e)x$, so the
    second integral evaluates to $-\log e$, yielding \cref{eq:KLTrateconst}.
\end{proof}
\else
\fi

\section{Neural-Network-Based Compression}
\label{sec:neural}
\begin{figure}
\centering
\includegraphics{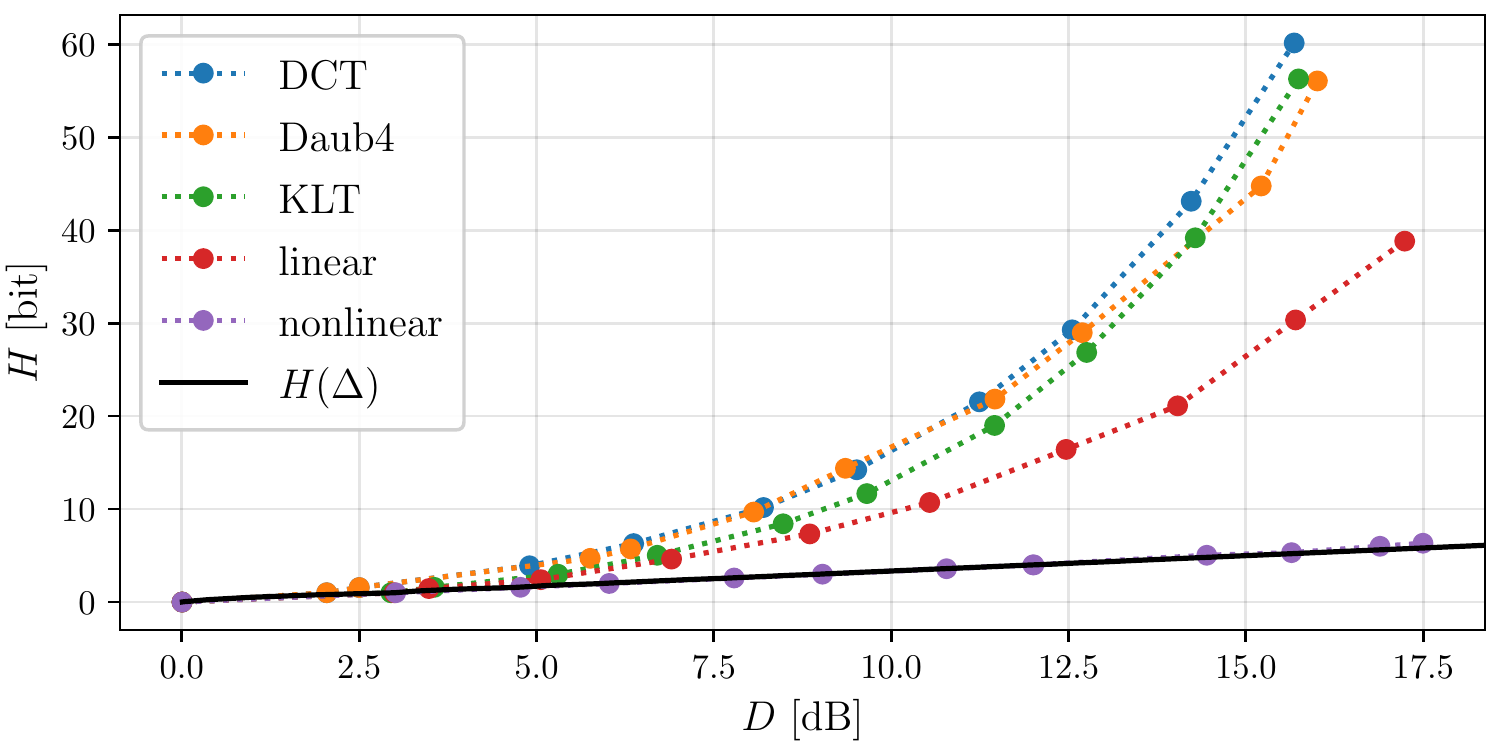}
\includegraphics{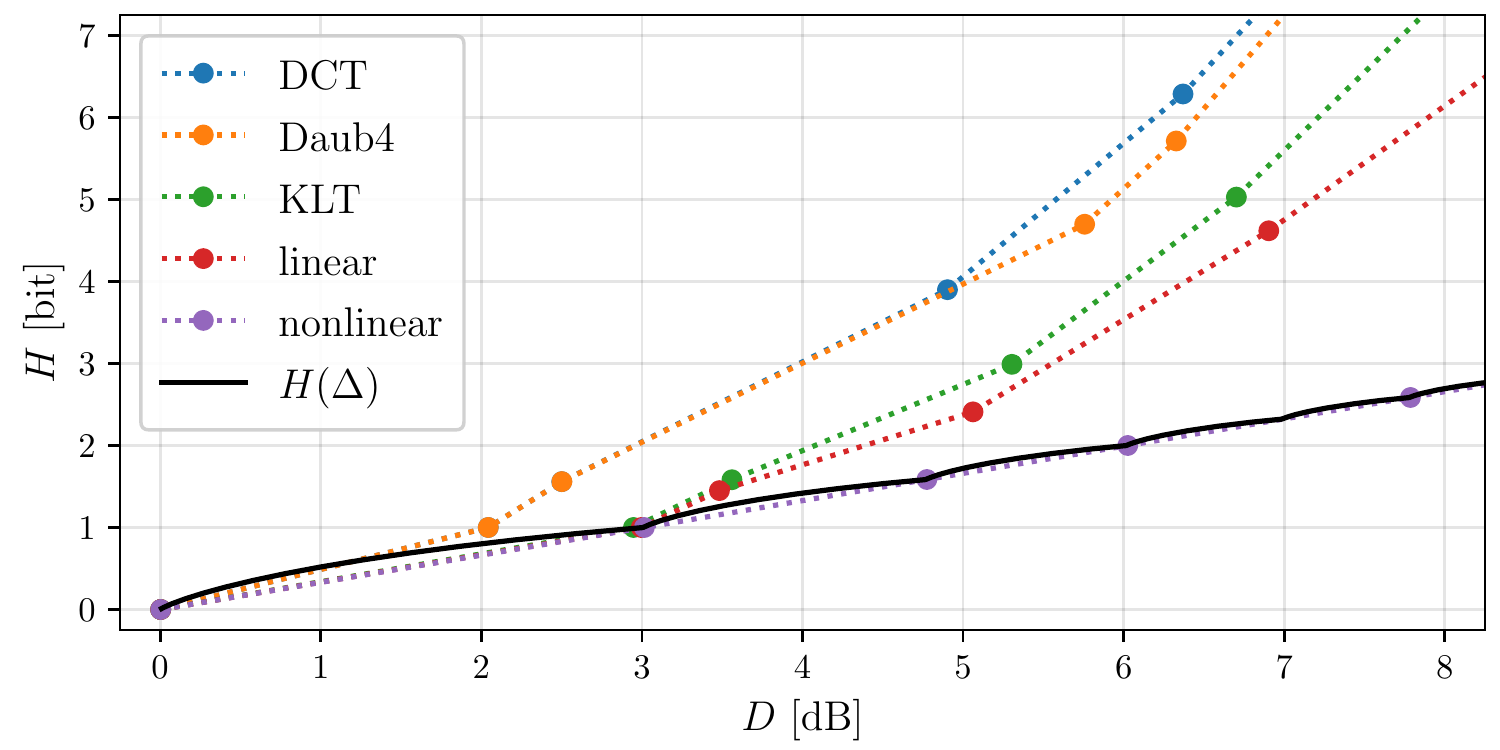}
\caption{Empirical entropy--distortion plots for transform codes constrained to discrete cosine transform (DCT), Daubechies 4-tap wavelet (Daub4), Karhunen--Loève transform (KLT), arbitrary linear transforms, and nonlinear transforms implemented by ANNs. We also plot the entropy--distortion function of the source. The bottom panel shows the same data, zoomed in to the low-rate regime.}
\label{fig:entropy_distortion}
\end{figure}

\begin{figure}
\centering
\includegraphics{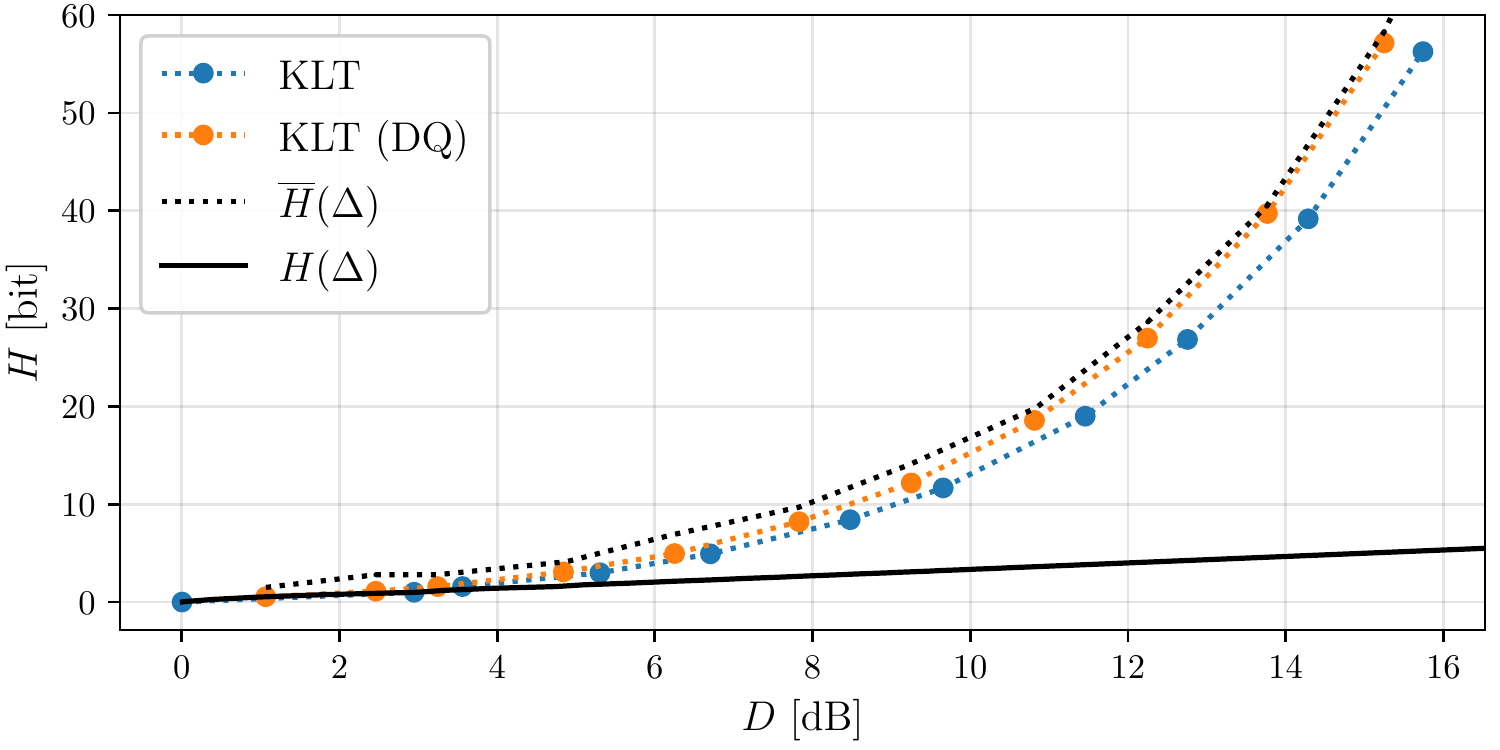}
\caption{Empirical entropy--distortion plot for transform codes constrained to Karhunen--Loève transform (KLT) with/without dithered quantization (DQ). We also plot the entropy--distortion function of the source, and the bound developed in \eqref{eq:KLTrateconst}.}
\label{fig:entropy_distortion_klt}
\end{figure}

We compare the performance of experimentally-trained ANNs against
the optimal entropy-distortion tradeoff in \cref{thm:optimal} and various
linear-transform-based schemes.
We follow the approach summarized by \citet{BaChMiSiJo20}.
To represent the sawbridge digitally, we sample $t$ at 1024 equidistant points between 0 and 1; thus, each realization is represented as a 1024-dimensional vector.
Due to this discretization, only a finite number of realizations are possible, and we took care to keep this number high enough such that none of the transform codes were able to exploit it.
We optimize over three sets of model parameters: the weights and biases of an analysis ($f$) and synthesis ($g$) transform, both of which are represented by ANNs, as well as the parameters of a non-parametric entropy model $p$.  Note that we do not assume that the analysis and synthesis transforms form exact inverses of each other; not only would this be difficult to enforce with ANNs, it is also not necessarily optimal.
We employ ANNs of three layers, with 100 units each (except for the last), and leaky ReLU as an activation function (except for the last), for each of the transforms.  We perform uniform scalar quantization of the transform coefficients, and use an entropy model that assumes independence between each of the latent dimensions; i.e., each transform coefficient is assumed to be encoded separately.
The objective is to minimize the Lagrangian in \cref{eq:Lagrangian}, as in \citet{BaChMiSiJo20}.

To make the comparison with linear transform coding schemes fair, we optimize linear transforms using the same methodology (linear transforms are special cases of ANNs, with just one layer).  For a comparison with specified orthonormal transforms, we express the analysis transform as the composition of a fixed orthonormal matrix with a trainable diagonal scaling matrix (and analogously for the synthesis transform, in reversed order).  The scaling enables uniform quantization with different effective step sizes in each latent dimension.

Empirical results are plotted in \cref{fig:entropy_distortion}.
Each point in the plots represents the outcome of one individual optimization of the Lagrangian in \cref{eq:Lagrangian} with a particular predetermined value of $\lambda$, and with a predetermined constraint on the transforms (as described above).
We spaced $\lambda$ logarithmically in order to cover a wide range of possible trade-offs.
Both $H$ and $D$ are computed as empirical averages over $10^7$ source realizations, where $H$ represents the cross-entropy between the fitted entropy model $p$ and the empirical distribution of the coefficients, and $D$ is mean squared error across $t$.
The top panel in the figure illustrates that the growth rate of the
entropy of the linear transform codes is highly suboptimal, especially when the linear transform is further constrained to coincide with a fixed orthogonal transform such as the DCT, KLT, or a wavelet transform.
In the bottom panel, we observe that the nonlinear transform code is
essentially optimal. Since the transform codes are optimized for the Lagrangian, they settle into the ``kinks'' of the entropy--distortion function, as discussed in \cref{sec:optimal} (note that not all of the kinks are occupied; this is due to the pre-determined spacing of $\lambda$).
As predicted by the theory, we find that a transform code with a linear analysis transform and a nonlinear synthesis transform, optimized with the same method, performs the same as the nonlinear code plotted in the figure (not shown).

Inspection of the latent-space activations as well as the entropy model reveal that linear transforms require successively larger numbers of latent dimensions as the rate increases.
Nonlinear codes, on the other hand, use only a small number of latent dimensions (additional dimensions allowed for by the experimental setup are collapsed to zero-entropy distributions by the optimization procedure), implying that they successfully discover the low-dimensional structure of the source.
In fact, we find that when constraining the latent space to a single dimension, the nonlinear codes do just as well. A simple explanation for why they may choose to use more dimensions in some cases is that an entropy model consisting of a single uniform distribution over a discrete set of states can be factorized into a product of multiple uniform distributions, as long as the number of states is not prime.

To illustrate the optimal KLT-constrained code from \cref{sec:linear}, and to verify that our optimization methodology is valid for linear codes, we plot $\overline H(\Delta)$ from \cref{eq:KLTfinitek} and the optimal tradeoff, $H(\Delta)$, along with empirical results for a KLT-constrained code with and without dithered quantization, in \cref{fig:entropy_distortion_klt}. Note that $\overline H(\Delta)$ and the curve for dithered quantization are quite close, especially at high rates.

\section*{Acknowledgment}
The first author wishes to thank Jingsong Lin for his contributions
to the early phases of this work. He was supported
by the US National Science Foundation under grants CCF-1617673 and
CCF-2008266 and the US Army Research Office under grant W911NF-18-1-0426.

\printbibliography

\end{document}